\documentclass[hidelinks,twocolumn,twoside]{IEEEtran} 
\usepackage{times}

\usepackage{graphicx}
\usepackage[font=footnotesize]{caption}
\usepackage{amsmath}
\usepackage{amsfonts}
\usepackage{amssymb}
\usepackage{amsthm}
\usepackage{tabularx}
\usepackage{mathrsfs} 
\usepackage{times}
\usepackage{cite}

\usepackage{subfig}

\usepackage{stfloats}
\usepackage{float}
\usepackage{url}
\usepackage{hyperref}
\newtheorem{theorem}{Theorem}
\newtheorem{definition}{Definition}
\newtheorem{lemma}{Lemma}
\newtheorem{corollary}{Corollary}
\newtheorem{remark}{Remark}
\usepackage[nolist,nohyperlinks]{acronym}
\usepackage[numbers,sort&compress]{natbib}
\renewcommand{\vec}[1]{\mathbf{#1}}

\DeclareGraphicsExtensions{.png,.eps,.ps,.pdf}

\usepackage{tikz}
\usetikzlibrary{arrows,	
				shapes, 
				positioning, 
				fit,	
                backgrounds, 
                mindmap,
                petri,
				intersections, 
				external %
				}
\definecolor{emerald}{RGB}{69,155,61}
\definecolor{gold}{RGB}{244,216,51}
\definecolor{pink}{RGB}{235,44,206}
\usepackage{verbatim}

\begin{acronym}
\acro{SNR}{signal-to-noise ratio}
\acrodefplural{SNR}[SNRs]{signal-to-noise ratios}
\acro{RV}{random variable}
\acrodefplural{RV}[RVs]{random variables}
\acro{LOS}{line-of-sight}
\acrodefplural{LoS}{Line-of-sight}
\acro{EH}{Energy Harvesting}
\acro{PDF}{probability density function}
\acro{CDF}{cumulative distribution function}
\acrodefplural{CDF}[CDFs]{cumulative distribution functions}
\acro{IoT}{Internet of Things}
\acro{PB}{Power Beacon}
\acrodefplural{PB}[PBs]{Power Beacons}
\acro{WPC}{Wireless Powered Communications}
\acro{WPT}{Wireless Power Transmission}
\acro{RFID}{Radio Frequency IDentification}
\acro{MC}{Monte Carlo}
\acro{PLS}{Physical Layer Security}
\acro{SWIPT}{Simultaneous Wireless and Information Power Transfer}
\end{acronym}

\tikzstyle{int}=[draw, fill=cyan!20, minimum size=2em]
\tikzstyle{int_blue}=[draw, fill=blue!20, minimum size=2em]
\tikzstyle{int_green}=[draw, fill=green!20, minimum size=2em]
\tikzstyle{int_red}=[draw, fill=red!20, minimum size=2em]
\tikzstyle{init} = [pin edge={to-,thin,black}]

\begin{document}
\title{\LARGE{Capacity of Backscatter Communication\\ under Arbitrary Fading Dependence}}
\author{F.\,Rostami Ghadi, F.\,J. Mart\'in-Vega and  F.\,J. L\'opez-Mart\'inez
\thanks{Manuscript received May xx, 2021; revised XXX. This work was funded by the Spanish Government (Ministerio de Ciencia e Innovaci\'on) through grant TEC2017-87913-R, and by Junta de Andalucia (project P18-RT-3175) and PAIDI 2020.}
\thanks{The authors are with Dpto. Ingenieria de Comunicaciones, Universidad de Malaga, 29071 Malaga, Spain. (e-mail: $\rm farshad@ic.uma.es$, $\rm fjmvega@ic.uma.es$, $\rm fjlopezm@ic.uma.es$)}
}
\maketitle
\begin{abstract}
We analyze the impact of arbitrary dependence between the forward and backward links in backscatter communication systems. Specifically, we quantify the effect of positive and negative dependence between these fading links on channel capacity, using Copula theory. The benefits of this approach are highlighted over the classical framework of linear dependence based on Pearson's correlation coefficient, which is also analyzed. Results show that for a fixed transmit power budget, capacity grows with positive dependence as well as with fading severity in the low signal-to-noise ratio (SNR) regime. Conversely, fading dependence becomes immaterial in the high SNR regime.
\end{abstract}


\section{Introduction}
Emerging technologies such as Radio Frequency Identification (RFID) systems and the Internet of Things (IoT) have led to significant attention being paid to backscatter communication (BC) in the recent years. In contrast to conventional radio frequency (RF) communication systems, BC exploits reflected power to modulate the signals, which leads to low power consumption and cost \cite{stockman1948communication}. One of the distinctive characteristics of BC is the non-negligible correlation between the forward (i.e., transmitter-to-tag) and backward (i.e., tag-to-reader) links \cite{griffin2007link,alhassoun2019theoretical}
, especially when the transmitter and the reader are co-located. From a communication-theoretic perspective, the equivalent channel observed by the receiver is built as the product of two correlated random variables (RVs), which largely complicates the performance evaluation of BC systems. Hence, related literature is scarce \cite{zhang2018backscatter,bekkali2014performance,gao2016performance} and mostly focused on outage-based and error rate metrics. Only recently, the effect of link correlation in the capacity of BC was addressed in \cite{matez2020effect} for the case of Rayleigh fading, suggesting that correlation could be beneficial in some instances to improve performance. 

The role of general dependence structures beyond linear correlation is gaining momentum in the wireless community, since the consideration of potentially negative dependences between the different sources of randomness is known to have an impact on system performances \cite{besser2020copula,ghadi2020copula,besser2020copula1,ghadi2020copula1}. In this regard, Copula theory is a flexible procedure that allows for incorporating both positive and negative dependence structures between RVs, which is not always possible when conventional linear correlation is used to describe the dependency.

With all the above considerations, several practical questions in BC remain unanswered to date: (\emph{i}) What is the effect of more general dependence structures between the forward and backward links?; (\emph{ii}) How does fading severity affect performance of BC in these scenarios? In this work, we combine Copula theory with conventional statistical techniques to analyze a general BC scenario on which any arbitrary dependence pattern can be considered. The case of Nakagami-$m$ fading model is used to exemplify how fading severity affects system performance. Our general formulation also allows to derive tractable asymptotic expressions for the capacity of BC systems in the low/high signal-to-noise ratio (SNR) regimes.

\vspace{-3mm}
\section{System Model}
We consider a general BC system 
consisting of a transmitter, a passive tag and a reader, where the forward (i.e., transmitter-to-tag) channel $h_f$ and the backscatter (i.e., tag-to-reader) channel $h_b$ are dependent RVs. For simplicity and without loss of generality, the different agents are assumed to be single-antenna devices. Thus, the instantaneous received signal power at the reader can be compactly expressed as \cite{bekkali2014performance}: $P_R=P_TL_Tg_fg_b=\overline{P}_Rg_fg_b$,
where $P_T$ denotes the transmit power, and $L_T$ encapsulates a number of effects including polarization losses, path losses
and the tag's power transfer efficiency among other impairments.
The terms $g_i=|h_i|^2$ for $i\in\{f,b\}$ are the fading power channel coefficients associated to the forward and backward links, respectively. The fading coefficients are assumed normalized, meaning that $\mathbb{E}[g_f]=\mathbb{E}[g_b]=1$, where $\mathbb{E}[\cdot]$ denotes the expectation operator. Therefore, $\overline{P}_R$ represents the average receive power when the forward and backward links are independent. The instantaneous SNR at the reader is given by
\begin{align}
\label{eqsys}
\gamma=\tfrac{\overline{P}_R g_fg_b}{N_0}=\hat{\gamma}g_fg_b,
\end{align}
where $N_0$ is the noise power and $\hat{\gamma}=\frac{\overline{P}_R}{N_0}$ is the average SNR at the receiver side \emph{in the absence of correlation}, i.e., when $\mathbb{E}[g_fg_b]=\mathbb{E}[g_f]\mathbb{E}[g_b]=1$. However, in the general case of arbitrary dependence between $g_f$ and $g_b$, the expected value of the product channel will be determined by the underlying joint distribution of $g_f$ and $g_b$. 
 \vspace{-0.3cm}
\section{SNR distributions}
\subsection{Preliminary definitions}
In order to determine the distribution of $\gamma$ in the general case, we now briefly review some basic definitions and properties of the two-dimensional Copulas \cite{nelsen2007introduction}.

\begin{definition}[Two-dimensional Copula]
	Let $\vec{X}=(X_1,X_2)$ be a vector of two RVs with marginal cumulative distribution functions (CDFs) $F_{X_j}(x_j)=\Pr(X_j\leq x_j)$ for $j\in\{1,2\}$, respectively. The relevant bivariate CDF is defined as:
	\begin{align}
	F_{X_1,X_2}(x_1,x_2)=\Pr(X_1\leq x_1,X_2\leq x_2),
	\end{align}
	then, the Copula function $C(u_1,u_2)$ of the random vector $\vec{X}$ defined on the unit hypercube $[0,1]^2$ with uniformly distributed RVs $U_j:=F_{X_j}(x_j)$ for $j\in\{1,2\}$ over $[0,1]$ is given by
	\begin{align}
	C(u_1,u_2)=\Pr(U_1\leq u_1,U_2\leq u_2).
	\end{align}
\end{definition}
\begin{theorem}[Sklar's theorem]\label{thm-sklar}
	Let $F_{X_1,X_2}(x_1,x_2)$ be a joint CDF of RVs with marginals $F_{X_j}(x_j)$ for $j\in\{1,2\}$. Then, there exists one Copula function $C(\cdot,\cdot)$ such that for all $x_j$ in the extended real line domain $\mathbb{R}$,%
	\begin{align}\label{sklar}
	F_{X_1,X_2}(x_1,x_2)=C\left(F_{X_1}(x_1),F_{X_2}(x_2)\right).
	\end{align}
\end{theorem}
\begin{corollary}\label{col-joint}
	Applying the chain rule to Theorem \ref{thm-sklar}, the joint probability density function (PDF) $f_{X_1,X_2}(x_1,x_2)$ is given as:
	\begin{align} \nonumber
	f_{X_1,X_2}(x_1,x_2)=f_{X_1}(x_1)f_{X_2}(x_2)c\big(F_{X_1}(x_1),F_{X_2}(x_2)\big),
	\end{align}
	where $c\big(F_{X_1}(x_1),F_{X_2}(x_2)\big)=\frac{\partial^2 C(F_{X_1}(x_1),F_{X_2}(x_2))}{\partial F_{X_1}(x_1)\partial F_{X_2}(x_2)}$  is the Copula density function and $f_{X_j}(x_j)$ for $j\in\{1,2\}$ are the marginal PDFs, respectively.
\end{corollary}
\begin{theorem}[Fr\'echet-Hoeffding bounds]
\label{theo:Copula bounds}
For any Copula function $C:[0,1]^2 \mapsto [0,1]$ and any $(u_1, u_2) \in [0,1]^2$, the following bounds hold:
$C^{-} \prec C \prec C^{+}$; where $C_1 \prec C_2$ if $C_1(u_1, u_2) \leq C_2(u_1, u_2) \; \forall (u_1,u_2)\in[0,1]^2$, and
\begin{align}\label{eq:C bounds}
& C^{-}(u_1, u_2) = \max(u_1 + u_2 - 1, 0) \nonumber \\
& C^{+}(u_1, u_2) = \min(u_1, u_2)
\end{align}
\end{theorem} 
The upper and lower 
bounds model extreme dependence structures. If $C=C^{-}$ the pair of RVs are said to be countermonotonic, whereas $C=C^{+}$ means that both RVs are comonotonic. 
The Copula function for the independent case $C^{\perp}(u_1,u_2)=u_1 \cdot u_2$ defines the limit between positive and negative dependence. Let us assume two Copulas that verify: 
$C^{-} \prec C_1 \prec C^{\perp} \prec C_2 \prec C^{+}$. Then, $C_1$ models a negative dependence and $C_2$ a positive dependence. 
\vspace{-3mm}
\subsection{General dependence}
Since the considered fading channels are correlated, the distribution of the SNR $\gamma$ is that of the product of two correlated RVs. For this purpose, we exploit the following theorem to determine the PDF of $f_\gamma(\gamma)$ \cite{ly2019determining}.
\begin{theorem}\label{thm-product}
	Let  $\textbf{X}=(X_1,X_2)$ be a vector of two absolutely continuous RVs with the corresponding Copula $C$ and CDFs $F_{X_j}(x_j)$ for $j\in\{1,2\}$. Thus, the PDF of $Y=X_1X_2$ is:
	\begin{align}
	f_Y(y)=\int_{0}^{\infty}c\left(u,F_{X_2}\left(\tfrac{y}{F_{X_1}^{-1}(u)}\right)\right)
	\tfrac{f_{X_2}\left(\tfrac{y}{F_{X_1}^{-1}(u)}\right)}{|F_{X_1}^{-1}(u)|} du,\label{thm-prod}
	\end{align}
	where $F^{-1}_{X_1}(.)$ is an inverse function of $F_{X_1}(.)$ and $c(.)$ denotes the density of Copula $C$.
\end{theorem}
\begin{proof}
	See \cite{ly2019determining}.
\end{proof}
\begin{corollary}\label{col-pdf}
	The PDF of $\gamma$  in the general dependence case of fading channels is given by
\begin{align}
f_\gamma(\gamma)=&\int_{0}^{\infty}c\left(u_1,F_{G}\left(\frac{\gamma}{\hat{\gamma}g_f}\right)\right)
\tfrac{f_{G}(g_f)}{\hat{\gamma}g_f} f_{G}\left(\frac{\gamma}{\hat{\gamma}g_f}\right)dg_f\label{gen-snr}
\end{align}
where $u_1=F_{G}\left(g_f\right)$.
\end{corollary}
\begin{proof}
Let $y=g_fg_b$ and $u=F_{G}(g_f)$ in Theorem \ref{thm-product}, and using the fact that $f_{\gamma}(\gamma)=\frac{1}{\hat{\gamma}}f_{Y}\big(\frac{\gamma}{\hat{\gamma}}\big)$ the proof is completed.
\end{proof}
Corollary \ref{col-pdf} holds for any arbitrary choice of fading distributions as well as Copula functions. For exemplary purposes, we will assume that the forward and backscatter fading channel coefficients (i.e., $h_f$ and $h_b$) follow the Nakagami-$m$ distribution \cite{zhang2018backscatter}. 
Hence, the corresponding fading power channel coefficients $g_i$ for $i\in\{f,b\}$ are dependent Gamma RVs. 
%
\begin{corollary}
\label{cor:lin-corr}
The Pearson's correlation coefficient, $\rho$, of a given pair of correlated Gamma RVs can be expressed in terms of the Copula function that models their dependence as:
\begin{align}
\label{eq:lin-corr}
\rho = m  \int \limits_{g_f \geq 0} \int \limits_{g_b \geq 0}
	\Big(C(u_1, u_2) -u_1 \cdot u_2 \Big) dg_f dg_b
\end{align}
where $u_1=F_G(g_f)$ and $u_2=F_G(g_b)$.
\begin{proof}
The proof comes after expressing the covariance of two positive RVs in terms of their joint CDF and adding the expression of the  variance of Nakagami-$m$ fading. 
\end{proof}
\end{corollary}

In view of \eqref{eq:lin-corr}, we note that Fr\'echet-Hoeffding bounds lead to the bounds of Pearson's correlation. While Copulas allow to model dependency structures other than linear, \eqref{eq:lin-corr} captures the linear part of the dependency between the RVs. 
The following lemma represents the bounds of linear correlation that can be achieved for Gamma marginals. 

\begin{lemma}
\label{lem:RhoBounds}
The upper and lower bounds of linear correlation for a pair of correlated RVs, with Gamma (Nakagami-m fading) marginals are expressed as follows:
\begin{align}
\label{eq:RhoBounds}
\rho^{+} &= 2 
- \tfrac{\Gamma \left( 2(m+1) \right)  {}_2F_1\left(1; -(1+m); 1+m; -1 \right)}{\Gamma^2(m) \cdot m^2 \cdot 2^{2 m}}\\
\rho^{-} &= m \int_{g_f \geq 0} \Big( \chi  \left( g_f, \omega \right) -  \kappa  \left( g_f, \omega \right) 
		 + \kappa \left( g_f, 0\right)  \Big) dg_f 
\end{align}
where $\omega=F_G^{-1} \left( \overline{F}_G(g_f)\right)$, $\overline{F}_G(g)$ is the complementary CDF of $G$, ${}_2F_1(\cdot)$ is Gauss hypergeometric function, $\Gamma(\cdot)$ is the Gamma function and $\chi(g_f,g_b)$ and $\kappa(g_f,g_b)$ are:
\begin{align} 
\chi(g_f, g_b) &= \tfrac{\left(\Gamma(m+1, g_b m)-g_b m \Gamma(m, g_b m)\right)}
	{m \Gamma(m)^{2} (\Gamma(m)-\gamma(m, g_f m))^{-1}} \nonumber \\
\kappa(g_f, g_b) &= \tfrac{\left(g_b m \Gamma(m)-g_b m \Gamma(m, g_b m)+\Gamma(m+1, g_b m)\right)} 
	{m \Gamma(m)^{2} \left(\gamma(m, g_f m) \right)^{-1}} \nonumber 
\end{align}
\begin{proof}
The proof is given in Appendix \ref{proof:RhoBounds}. 
\end{proof}
\end{lemma}

We will now exemplify how the key performance metrics of interest can be characterized in closed-form for the specific case of the FGM Copula, which is defined below. This choice is justified because it allows to capture \emph{both} positive and negative dependencies between the RVs while offering good mathematical tractability. 
As we will later see, the use of FGM Copula will be enough for our purposes of determining the effect of negative dependence between BC links.
\begin{definition}\label{def-fgm}[FGM Copula] The bivariate FGM Copula with dependence parameter $\theta\in[-1,1]$ is defined as:
	\begin{align}\label{fgm}
	C_F(u_1,u_2)=u_1u_2(1+\theta(1-u_1)(1-u_2)),
	\end{align}
	where $\theta\in[-1,0)$ and $\theta\in(0,1]$ denote the negative and positive dependence structures respectively, while $\theta=0$ \textit{always} indicates the independence structure.
\end{definition}
The permissive range of the linear correlation derived from the FGM Copula can be obtained by evaluating \eqref{eq:lin-corr} as 
$\rho \in [\rho(\theta=-1), \rho(\theta=1)] \subset [-1,1]$.
\begin{theorem}\label{thm-pdf}
	The PDF of $\gamma$ under correlated Nakagami-$m$ fading channels using the FGM Copula is given by
		\begin{align}\nonumber
	&f_\gamma^{\rm FGM}(\gamma)=\mathcal{B}\bigg(\gamma^{m-1}K_{0}\left(\zeta\sqrt{\gamma}\right)+\theta\bigg[\gamma^{m-1}K_{0}\left(\zeta\sqrt{\gamma}\right)\\\nonumber
	&-\sum_{k=0}^{m-1}\frac{2^{2-\frac{k}{2}}m^k}{\hat{\gamma}^{\frac{k}{2}}k!} \gamma^{\frac{k}{2}+m-1}
	K_{k}\left(\zeta\sqrt{2\gamma}\right)\\
	&+\sum_{k=0}^{m-1}\sum_{n=0}^{m-1}\frac{4m^{k+n}}{\hat{\gamma}^{\frac{k+n}{2}}k!n!}\gamma^{\frac{k+n}{2}+m-1}K_{n-k}\left(2\zeta\sqrt{\gamma}\right)\bigg]\bigg),\label{col-pdf-sim}
	\end{align}
	where $\mathcal{B}=\frac{2{m}^{2m}}{\hat{\gamma}^m\Gamma(m)^2}$, $\zeta=\frac{2m}{\sqrt{\hat{\gamma}}}$ and $K_v(.)$ is the modified Bessel function of the second kind and order $v$.
\end{theorem}
\begin{proof}
	The details of proof are in Appendix \ref{app-pdf}.
	\end{proof}

\vspace{-2mm}
\subsection{Linear dependence}
\label{seclin}
In the case of linear dependence, the PDF of $\gamma$ under correlated Nakagami-$m$ fading can be derived from Kibble's bivariate Gamma distribution \cite{AlouiniBook}, yielding
\begin{align}
f_{\gamma}^{\rm lin}(\gamma)=\mathcal{D}\gamma ^{\frac{m-1}{2}}I_{m-1}\left(\mathcal{A} \sqrt{\rho}\right) K_0\left(\mathcal{A}\right),\label{linear-pdf}
\end{align}
where $\overline{\gamma}=\mathbb{E}[\gamma]=\hat\gamma \left(\frac{m+\rho}{m}\right)$ is the average SNR at the receiver side, $\rho=\frac{\text{cov}(g_fg_b)}{\sqrt{\text{var}[g_f]\text{var}[g_b]}}$, is the power correlation coefficient, $\mathcal{D}=\frac{2 m^{m+1} \left(\frac{m \overline{\gamma}}{m+\rho }\right)^{-\frac{m+1}{2}}}{(1-\rho ) \rho ^{\frac{m-1}{2}} \Gamma (m)}$, 
$\mathcal{A} = \tfrac{2 m \sqrt{\gamma }}{(1-\rho ) \sqrt{\frac{m \overline{\gamma}}{m+\rho }}}$ and $I_v(.)$ is the modified Bessel function of the first kind and order $v$. We note that \eqref{linear-pdf} does not support negative correlation values.
\section{Average Capacity}
In this section, we compute the average capacity (per bandwidth unit) for the system model under consideration, as:
\begin{align}
\overline{\mathcal{C}}\left[\text{bps/Hz}\right]\overset{\Delta}{=}\int_{0}^{\infty}\log_2\left(1+\gamma\right)f_\gamma\left(\gamma\right)d\gamma.\label{capacity}
\end{align}
\subsection{General dependence}
In the most general case, i.e., for an arbitrary choice of Copula and fading distributions, the average capacity of a BC system can be computed by plugging \eqref{gen-snr} into \eqref{capacity}.
We will now pay attention to the specific examples considered in the previous section, i.e., when using the FGM Copula and under linear correlation.
	\begin{theorem}
	The average capacity under correlated Nakagami-$m$ fading channels using the FGM Copula can be evaluated as
	\begin{align}\nonumber
	&\overline{\mathcal{C}}=\mathcal{B}\int_{0}^{\infty}\bigg(\gamma^{m-1}K_{0}\left(\zeta\sqrt{\gamma}\right)+\theta\bigg[\gamma^{m-1}K_{0}\left(\zeta\sqrt{\gamma}\right)\\\nonumber
	&-\sum_{k=0}^{m-1}\frac{2^{2-\frac{k}{2}}m^k}{\hat{\gamma}^{\frac{k}{2}}k!} \gamma^{\frac{k}{2}+m-1}
	K_{k}\left(\zeta\sqrt{2\gamma}\right)\\
	&+\sum_{k=0}^{m-1}\sum_{n=0}^{m-1}\frac{4m^{k+n}}{\hat{\gamma}^{\frac{k+n}{2}}k!n!}\gamma^{\frac{k+n}{2}+m-1}K_{n-k}\left(2\zeta\sqrt{\gamma}\right)\bigg]\bigg)d\gamma.\label{capacity-copula}
	\end{align}
	\end{theorem}
The above integral can be expressed, in closed-form, in terms of the hypergeometric, digamma, and gamma functions, although such expression is not included here for the sake of compactness. Besides, it can be easily evaluated numerically. 

\subsection{Linear dependence}
The average capacity under correlated Nakagami-$m$ fading channels assuming linear correlation can be evaluated by plugging \eqref{linear-pdf} into \eqref{capacity}. This integral can be solved in terms of the Fox-H function following a similar rationale as in \cite{matez2020effect}. However, it can also be easily evaluated numerically. 

\subsection{Asymptotic results}
Even though the exact capacity can be evaluated numerically, we are interested in examining the asymptotic behavior 
in the high and low SNR regime. We first consider the case with linear correlation, although we will later see that they can be extrapolated to the case of arbitrary dependence.

\begin{theorem}
	In the high SNR regime, the average capacity of a BC system can be approximated as
\begin{align}
\label{asy1}
\overline{\mathcal{C}}(\overline{\gamma})&\approx\log_2(\overline{\gamma})+\log_2(e)\left(2\psi(m)-\log(m(m+\rho))\right),\\
\overline{\mathcal{C}}(\hat{\gamma})&\approx\log_2(\hat{\gamma})-2\log_2(m)+2\psi(m)\log_2(e),
\label{asy2}
\end{align}
for the cases of fixed average receive SNR and fixed transmit power, respectively, with $\psi(\cdot)$ being the digamma function.
\end{theorem}
\begin{proof}
	The details of proof are in Appendix \ref{app-cap}.
\end{proof}
%

\begin{figure*}[t!]%
    \centering
    \subfloat[]{{\includegraphics[width=0.33\textwidth]{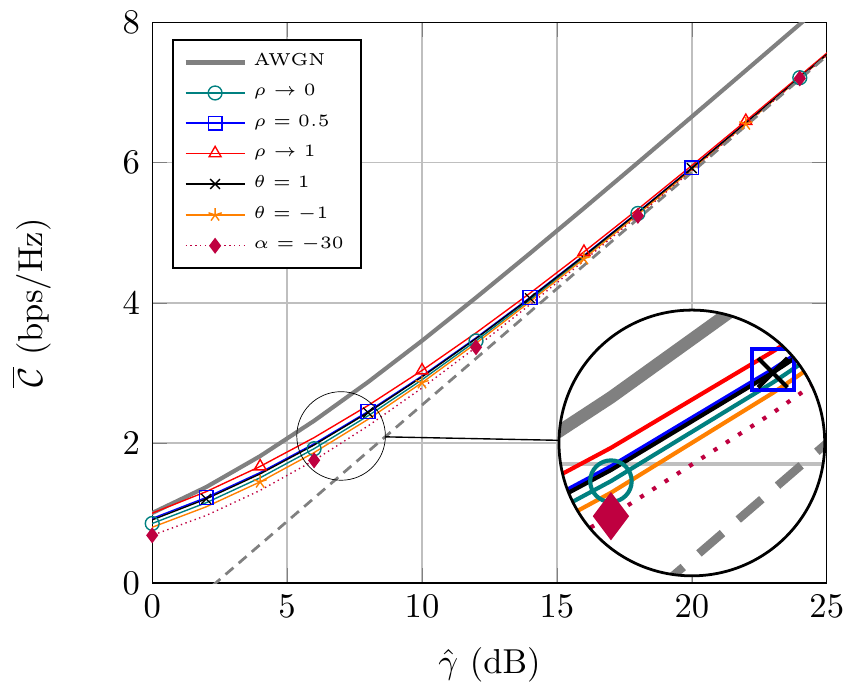} }}%
    \subfloat[]{{\includegraphics[width=0.33\textwidth]{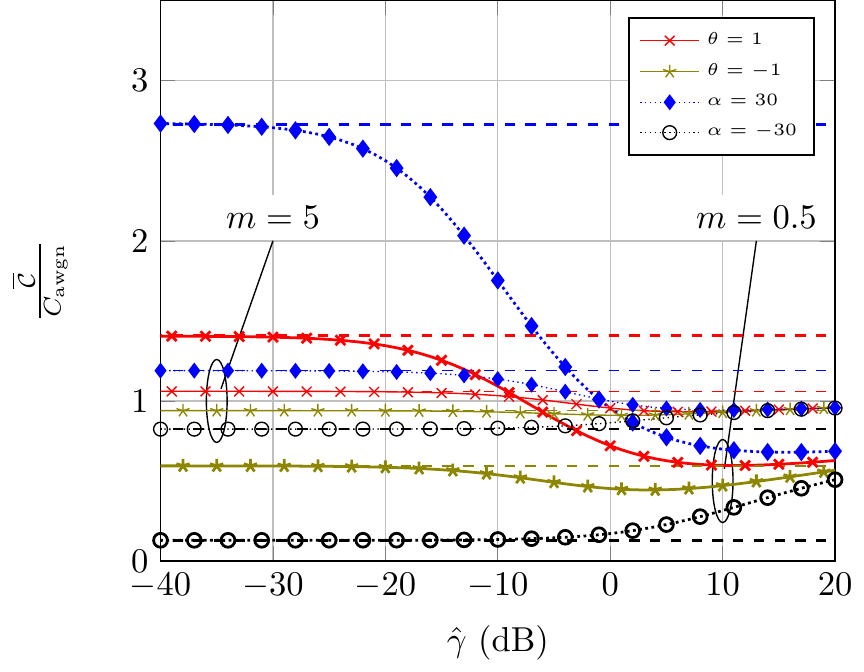} }}%
    \subfloat[]{{\includegraphics[width=0.33\textwidth]{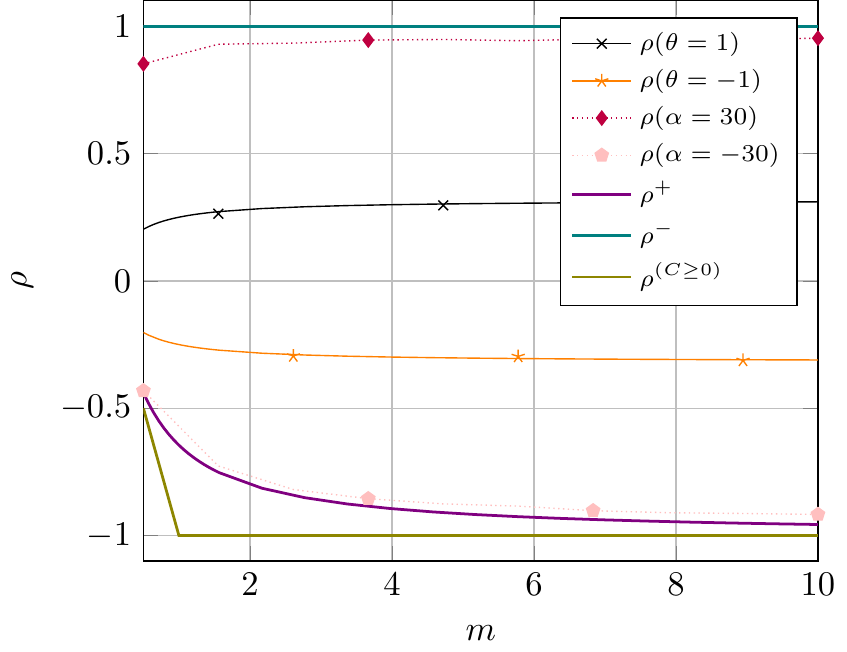} }}%
    \caption{Capacity metrics versus $\hat{\gamma}$ under different scenarios:
    (a) average capacity $\overline{\mathcal{C}}$ versus $\hat{\gamma}$ for $m=2$;
    (b) average capacity $\overline{\mathcal{C}}$ normalized to that of the AWGN case versus $\hat{\gamma}$ for $m=0.5$ (thin lines) and $m=5$ (thick lines); and
    (c) Pearson's correlation coefficient vs fading parameter $m$ for FGM and Frank Copulas and theoretic bounds.
Solid lines indicate theoretical expressions. Markers indicate MC simulations. Dashed gray line indicates the asymptotic expression. 
For $m=0.5$, $\rho_{\rm FGM}\left(\theta=\pm1\right)\approx \pm 0.2$; $\rho_{\rm Frank}\left(\alpha=\pm30\right)\approx \left[0.86,-0.43\right]$; whereas for $m=5$, $\rho_{\rm FGM}\left(\theta=\pm1\right)\approx \pm 0.3$; $\rho_{\rm Frank}\left(\alpha=\pm30\right)\approx \left[0.94,-0.88\right]$.
    }%
    \label{fig-all-results}%
\end{figure*}

The next two remarks provide theoretical insights about the system behavior in the high-SNR regime. 
\begin{remark}
In view of \eqref{asy1}, we see that the linear correlation decreases (increases) the average capacity in the high SNR regime for a fixed average receive SNR $\overline\gamma$ and positive (negative) dependence. This can be understood considering that an increment (decrement) on the correlation also increases (decreases) the variance of the product channel gains $g_f \cdot g_b$, which decreases (increases) the average capacity. 
\end{remark}
\begin{remark}
By inspection of \eqref{asy2}, we state that in the high-SNR regime with a fixed transmit power, the linear correlation coefficient does not affect average capacity, only the fading severity $m$. This is justified as follows. For a fixed $P_T$ in \eqref{eqsys} the $\hat\gamma$ is also fixed, and $\overline\gamma=\hat\gamma\mathbb{E}\left[g_fg_b\right]$ will vary depending on $\mathbb{E}\left[g_fg_b\right]$. Since $\mathbb{E}\left[g_fg_b\right]=\frac{m+\rho}{m}\lessgtr1$ for positive/negative dependence, the average SNR increases (decreases) as $\rho$ increases (decreases), which compensates the decrement (increment) of capacity related to the increase (decrease) on $\rho$.
\end{remark} 

\begin{lemma}
\label{lemma:asy-capacity-low}
	In the low-SNR regime, the average capacity of a BC system can be approximated as
\begin{align}
\overline{\mathcal{C}}(\hat{\gamma})\approx\log_2(e)\left(1+\tfrac{\rho}{m}\right)\hat{\gamma}.
\label{asy-capacity-low}
\end{align}
\end{lemma}
\begin{proof}
In the low SNR regime, $\log(1+\gamma)\approx\gamma$, which yields $\overline{\mathcal{C}}(\overline{\gamma})\approx\log_2(e)\overline{\gamma}$. Expressing $\overline\gamma$ in terms of $\hat\gamma$ completes the proof.
\end{proof}
\begin{remark}
Observation of lemma \ref{lemma:asy-capacity-low} demonstrates that, on the low SNR regime with a fixed $P_T$, the correlation increases the average capacity whereas it is immaterial for a fixed $\overline\gamma$.

\end{remark}
\begin{remark}
In view of lemma \ref{lemma:asy-capacity-low} we note that there is a communication theoretic bound for the correlation that depends on the fading severity. This bound is coherent with the fact that capacity cannot be negative and it can be expressed as follows: 
$\rho^{(C\geq0)} = - \mathrm{min}(1,m)$. 
\end{remark}

The asymptotic results for the average capacity in the high and low SNR regime have been derived from the linear correlation model in \eqref{linear-pdf}. However, there is a univocal relation between the linear correlation coefficient and dependence parameter of any arbitrary Copula as given by Corollary \ref{cor:lin-corr}.

\section{Numerical Results}
We now evaluate the theoretical expressions previously derived, which are double-checked in all instances with Monte Carlo (MC) simulations. For the case of negative dependence between the RVs, additional simulations are included using the Frank Copula. This allows us to extend the analytical results obtained with the FGM Copula to a wider range of values for negative dependence by means of Frank's Copula parameter $\alpha \in \mathbb{R} \setminus \{0\}$. 
In all figures, we assume a fixed transmit power $P_T$ and vary the fading parameter $m$ and the power correlation coefficient $\rho$.

Fig. \ref{fig-all-results}-(a) shows the behavior of the average capacity in terms of $\hat{\gamma}$ under correlated Nakagami-$m$ backscatter fading channels with $m=2$. It can be seen that 
the correlation effect is gradually eliminated in the high-SNR regime: the variance is increased (decreased) for positive (negative) dependence, which increases fading severity, although this effect is compensated by the increase (decrease) in average receive SNR $\overline\gamma$ when positive (negative) dependence is accounted for. 
Zooming into the range of lower SNR values, we see that positive (negative) dependence improves (worsens) the performance compared to the independent fading case.

In order to better understand the effect of fading dependence in the low SNR regime, we now normalize to average capacity to that of the AWGN case. 
In Fig. \ref{fig-all-results}-(b) it becomes evident that correlated fading provides a larger capacity under the positive dependence compared to the independent fading case as well as the absence of fading. Noteworthy, negative dependence structures are detrimental for capacity in the low-SNR regime, as stated by the curves obtained with the FGM (theory) and Frank (simulation) Copulas. We also see that this effect becomes more noticeable under a strong fading condition ($m=0.5$) than when a milder one ($m=5$) is considered. 

Since capacity in the low/high SNR regimes only depends on fading severity through $m$ and in the Pearson's correlation coefficient, we represent the latter as a function of $m$, together with the Fr\'echet-Hoeffding bounds in Fig. \ref{fig-all-results}-(c).
These bounds represent the maximum and minimum linear correlation that can be achieved with any dependence structure provided that the marginals are Gamma distributed. It should be highlighted that these bounds are not symmetric w.r.t. $\rho = 0$, but they tend to $\pm 1$ as the fading is less severe, i.e., as $m$ increases. 
We observe that the Frank Copula gets close to the Fr\'echet-Hoeffding bounds for $\alpha \pm 30$, and that exhibits a non-symmetric behavior. This is in coherence with Remark 4. We also see that correlation for the case with FGM Copula is now symmetric, since such Copula only is able to model weak dependences, and hence is distant to the Fr\'echet-Hoeffding bounds, thus not reaching the maximum permissive range of $\pm 1$ for $\rho$. 
%
%
\section{Conclusion}
We proved that the correlation between the forward and backward links has an impact on the capacity of BC affected by fading. To this end, we introduced a general framework based on Copula theory to model arbitrary dependence structures over both fading links. We proved that for a fixed transmit power, an increase in the correlation or an increase in the fading severity increases the average capacity in the low SNR regime, whereas in the high SNR regime the correlation does not affect the channel capacity and only fading severity affects this metric. 
Finally, we relied on Copula theory to obtain the bounds of linear correlation that can be achieved with a pair of Gamma distributed RVs, which support the capacity results here obtained.
\appendices

\section{Pearson's correlation bounds}\label{proof:RhoBounds}
The upper bound of linear correlation can be written by substituting 
$C^{+}$ on \eqref{eq:lin-corr}, which yields to a double integral whose integrand can be expressed as: $\min(u_1, u_2) - u_1 u_2$, with $u_1=F_G(g_f)$ and 
$u_2=F_G(g_b)$. Then, we express $\min(u_1,u_2) = u_1 \mathbf{1}(u_1 \geq u_2) + u_2 \mathbf{1}(u_1 < u_2)$, we split the double integral in two integrals and we apply the conditions imposed by the indicators over the integration limits. When we apply such conditions we use the following equality $\mathbf{1}(F_G(g_f) \geq F_G(g_b)) = \mathbf{1}(g_f \geq g_b)$ since CDF functions are non-decreasing functions. After these operations we realize that the two integrals are equivalent and thus we obtain the next expression:
\begin{align}
\rho^+ &= 2 m \int_{g_f \geq 0} \int_{g_b=0}^{g_f} F_G(g_b) \overline{F_G}(g_f) dg_b dg_f
\end{align}
Then, solving the integral for Gamma (Nakagami-$m$ fading) completes the proof. 

For the upper bound we follow a similar approach, expressing $C^- = \max(u_1 + u_2 - 1, 0)$ as the sum of two indicator functions. After applying the indicator functions over the integration limits we obtain the next expression:

\begin{align}
 \rho^- &= \int_0^\infty \Bigg[ \int^{\infty}_{g_b = F_G^{-1}\left(\overline{F}_G(g_f) \right)} 
	\Big( u_1 + u_2 - 1 - u_1 u_2 \Big) dg_b \nonumber \\
	&- \int^{F_G^{-1}\left(\overline{F}_G(g_f) \right)}_{g_b = 0} 
	\Big( u_1 u_2 \Big) dg_b \Bigg] dg_f
\end{align} 
Finally, computing the two inner integrals completes the proof. 

\section{Proof of Theorem \ref{thm-pdf}}\label{app-pdf}
By exploiting the Corollary \ref{col-joint} and Definition \ref{def-fgm}, the Copula density of the FGM can be obtained as
\begin{align}
c(u_1,u_2)=1+\theta\left((2u_1-1)(2u_2-1)\right)\label{fgm-c}
\end{align}
where $u_1=F_{G}(g_f)$ and $u_2=F_{G}\left(g_b\right)$. The marginal PDFs and CDFs for $g_f$ and $g_b$ are those of the Gamma distribution. 
%
Thus, the PDF of the product of fading power channel coefficients, $Y=G_f G_b$, can be determined as 
\begin{align} \nonumber
f_Y(y)=I_1+\theta(I_1-2I_2-2I_3+4I_4)
\end{align}
where
\begin{align} \nonumber
&I_1=\int_{0}^{\infty}\tfrac{1}{g_f}f_{G}(g_f)f_{G}\left(\tfrac{y}{g_f}\right)dg_f \\ \nonumber
&I_2=\int_{0}^{\infty}\tfrac{1}{g_f}f_{G}(g_f)f_{G}\left(\tfrac{y}{g_f}\right)F_{G}(g_f)dg_f\\\nonumber
&I_3=\int_{0}^{\infty}\tfrac{1}{g_f}f_{G}(g_f)f_{G}\left(\tfrac{y}{g_f}\right)F_{G}\left(\tfrac{y}{g_f}\right)dg_f \\\nonumber
&I_4=\int_{0}^{\infty}\tfrac{1}{g_f}f_{G}(g_f)f_{G}\left(\tfrac{y}{g_f}\right)F_{G}(y)F_{G}\left(\tfrac{y}{g_f}\right)dg_f
\end{align}
Finally, computing the above integrals and exploiting the fact that $f_{\gamma}(\gamma)=\frac{1}{\hat{\gamma}}f_{Y}\big(\frac{\gamma}{\hat{\gamma}}\big)$ completes the proof.

\section{High-SNR asymptotic capacity}\label{app-cap}
In the high SNR regime, the average capacity can be approximated as in \cite{yilmaz2012unified}
\begin{align}
\overline{\mathcal{C}}(\overline{\gamma})\approx\log_2(\overline{\gamma})+\log_2(e)\frac{d\mathcal{M}(n)}{dn}\Big\vert_{n=0},
\end{align}
where $\mathcal{M}(n)\overset{\Delta}{=}\frac{\mathbb{E}[{\gamma}^n]}{\overline{\gamma}^n}$ are the normalized moments of $\gamma$, which can be expressed from \eqref{linear-pdf} as
\begin{align}
\mathcal{M}(n)=\frac{\Gamma(m+n)^2 \, _2F_1(-n,-n;m;\rho)}{\left({m (m+\rho )}\right)^n\Gamma (m)^2}.\label{moment}
\end{align}
By repeatedly applying the derivative chain rule, the derivative of \eqref{moment} with respect to $n$ can be computed as in \cite{Laureano2016}. After some algebra, we have \eqref{asy1}.
 Finally, using the relationship between $\overline\gamma$ and $\hat\gamma$ in Section \ref{seclin}, \eqref{asy2} is obtained.

\bibliographystyle{IEEEtran}
\bibliography{cbs.bib}
\end{document}